\documentclass{amsart}

\usepackage[svgnames]{xcolor}
\usepackage{amsmath,amssymb,amsthm,doi,csquotes,booktabs,enumitem,graphicx,microtype,hyperref,tikz}
\usepackage[noabbrev,nameinlink,capitalize]{cleveref}
\usepackage[square,numbers]{natbib}
\usepackage{thmtools,thm-restate} 

\setlist[enumerate]{label=\textup{(\roman{*})}}
\usetikzlibrary{arrows,positioning,cd}

\newcommand*{\Z}{\mathbb{Z}}
\newcommand*{\ret}{\mathcal{R}}
\newcommand*{\lang}{\mathcal{L}}
\newcommand*{\cA}{\mathcal{A}}
\newcommand*{\cB}{\mathcal{B}}
\newcommand*{\from}{\colon}
\newcommand*{\emptyw}{\varepsilon}
\newcommand*{\Ext}{\operatorname{E}}
\newcommand*{\LExt}{\operatorname{E}^-}
\newcommand*{\RExt}{\operatorname{E}^+} 
\newcommand*{\ExtGraph}{\mathcal{E}}
\newcommand*{\RauzyGraph}{\Gamma}

\newcommand*{\card}{\#}

\newcommand*{\gen}[1]{\langle #1\rangle}

\theoremstyle{plain}
\newtheorem{theorem}{Theorem}
\newtheorem{proposition}[theorem]{Proposition}
\newtheorem{lemma}[theorem]{Lemma}

\theoremstyle{definition}

\theoremstyle{remark} 
\newtheorem{example}[theorem]{Example}

\begin{document}

\title{Algebraic characterization of dendricity}

\author[F. Gheeraert]{France Gheeraert}
\author[H. Goulet-Ouellet]{Herman Goulet-Ouellet}
\author[J. Leroy]{Julien Leroy}
\author[P. Stas]{Pierre Stas}

\begin{abstract}
    Dendric shift spaces simultaneously generalize codings of regular interval exchanges and episturmian shift spaces, themselves both generalizations of Sturmian words. One of the key properties enforced by dendricity is the Return Theorem. In this paper, we prove its converse, providing the following natural algebraic perspective on dendricity: \emph{A minimal shift space is dendric if and only if every set of return words is a basis of the free group over the alphabet}.
\end{abstract}

\maketitle

\section{Introduction}
\label{s:intro}

Sturmian words, introduced in \cite{Morse1940}, are binary sequences well-known for having numerous properties and characterizations (see the surveys in~\cite{book/Lothaire2002} and~\cite{book/Fogg2002}). Combinatorially, they are the simplest aperiodic words; dynamically, they encode irrational rotations on the circle; geometrically, they approximate lines of irrational slope; and number-theoretically, they are strongly related to continued fractions.

Among the many generalizations of Sturmian words, two are particularly studied: the episturmian shift spaces \cite{Glen2009} with a combinatorial definition, and the codings of regular interval exchanges \cite{Ferenczi2013} originating in the dynamical approach. Both are examples of the main protagonists of this paper: dendric shift spaces. They are combinatorially defined by restricting the context surrounding finite patterns. 

The study of dendric shift spaces begins in 2015 when a group of five researchers (including the third author of this paper) publish a series of papers in which they consider so-called \emph{tree sets}~\cite{Berthe2015,Berthe2015a,Berthe2015c}, which are precisely the languages of dendric shift spaces.
Dendricity generalizes some of the remarkable combinatorial and dynamical properties of Sturmian words, such as affine factor complexity~\cite{Berthe2015}, constant number of return words~\cite{Berthe2015}, at most $d/2$ ergodic measures where $d$ is the size of the alphabet~\cite{Damron2020}, and the existence of an $S$-adic characterization~\cite{pre/Gheeraert2022}.

In this paper, we focus on the algebraic aspects of dendric shift spaces, and more specifically on the following property known as the Return Theorem: \emph{In a minimal dendric shift space, every return set forms a basis of the free group}~\cite{Berthe2015}. This result is crucial to understand various aspects of dendric shift spaces, such as their dimension groups \cite{Berthe2021}, S-adic representations~\cite{Gheeraert2022,pre/Gheeraert2022}, Schützenberger groups~\cite{Almeida2016}, and stabilizers~\cite{Berthe2018}. Our contribution is a proof that the converse of the Return Theorem also holds, leading to the following characterization.

\begin{restatable}{theorem}{characThm}\label{t:characterization}
    Let $X$ be a minimal shift space over $\cA$. The following assertions are equivalent:
    \begin{enumerate}
        \item $X$ is dendric;
        \item the sets $\ret_X(w)$, $w \in \lang(X)$, are tame bases of the free group $F_\cA$;
        \item the sets $\ret_X(w)$, $w \in \lang(X)$, are bases of the free group $F_\cA$.
    \end{enumerate}
\end{restatable}

The statement of \cref{t:characterization} is reminiscent of Vuillon's characterization of Sturmian words stating that an infinite binary word is Sturmian if and only if every factor has \emph{exactly} two return words~\cite{Vuillon2001}. Observe that our condition is stronger than Vuillon's on the number of return words. This is unavoidable as cardinality alone is not sufficient to infer dendricity on alphabets with more than 2 letters~\cite{Balkova2008}.

\cref{t:characterization} provides a natural algebraic interpretation of dendricity, which should be relevant to establish closure properties of the family of dendric shift spaces. Moreover, \cref{t:characterization} and its proof emphasize once again the interactions between symbolic dynamics, combinatorics on words, and algebra, opening the door to further research at the intersection of these topics.

\section{Preliminaries}
\label{s:pre}

We briefly recall the main notions used in this paper. For more details, we refer to the following monographs: on combinatorics on word~\cite{book/Lothaire1997}; on the free group~\cite{book/Lyndon2001}; on symbolic dynamics~\cite{book/Lind1995}.

Let $\cA$ be a finite set that we call an {\em alphabet}, and let $\cA^*$ and $F_\cA$ respectively denote the free monoid (whose elements are called \emph{words}) and the free group generated by $\cA$. We let $\emptyw$ denote the neutral element of $\cA^*$, i.e., the empty word, and we set $\cA^+ = \cA^* \setminus \{\emptyw\}$. We naturally view $\cA^*$ as a subset of $F_\cA$, and thus we may consider the subgroup of $F_\cA$ generated by a given set of words $W \subseteq \cA^*$, which is then denoted $\gen{W}$.

We let $\cA^\Z$ denote the set of two-sided infinite words equipped with the product topology of the discrete topology over $\cA$. A \emph{shift space} is a closed subset $X\subseteq\cA^\Z$ invariant under the shift map $(x_n)_{n\in\Z}\mapsto (x_{n+1})_{n\in\Z}$. The \emph{language} of a shift space $X$ is the set
\begin{equation*}
    \lang(X) = \{ x_i\cdots x_{j} \mid x\in X, \ i\leq j\}.
\end{equation*}
A shift space is \emph{minimal} (with respect to the inclusion of shift spaces) precisely when $\lang(X)$ is \emph{uniformly recurrent}, meaning that for every $u\in \lang(X)$, there exists $n$ such that $u$ is a factor of every $w\in \lang(X) \cap \cA^n$.

Given a shift space $X$ and a word $w \in \lang(X)$, the set of \emph{(left) return words} to $w$ is
\[
    \ret_X(w) = \{r \in \cA^+ \mid rw\in \lang(X)\cap w\cA^*\setminus \cA^+w\cA^+\}.
\]

\begin{example}
    Let $\cA = \{a,b,c\}$. The \emph{Tribonacci shift space} is the shift space $X$ generated by the monoid morphism $\sigma\from a \mapsto ab, b \mapsto ac, c \mapsto a$, i.e., its language is the set of factors of $\sigma^n(a)$, $n \geq 0$.
    The return words to $aba$ are the words separating consecutive occurrences of $aba$ in the elements of $X$. This can be seen in the following prefix of $\sigma^5(a)$:
    \[
        abac|aba|abac|ab|abac.
    \]
    Hence the words $ab$, $aba$, and $abac$ are return words to $aba$. As a matter of fact, they are the only return words (see~\cite[Corollary 4.5]{Justin2000} for example). Moreover, they form a basis of the free group $F_{\cA}$, since
    \begin{equation*}
        a = (ab)^{-1}(aba), \quad b = (aba)^{-1}(ab)^2, \quad c = (aba)^{-1}(abac).
    \end{equation*}
\end{example}

For a shift space $X$ and a word $w\in\lang(X)$, we consider the sets of {\em extensions}
\begin{align*}
    \LExt_X(w) &= \{ a\in \cA \mid aw\in \lang(X) \},\\
    \RExt_X(w) &= \{ b\in \cA \mid wb\in \lang(X) \},\\ 
    \Ext_X(w) &= \{(a,b) \in \cA\times \cA \mid awb\in \lang(X)\}.
\end{align*}
The \emph{extension graph} of $w\in \lang(X)$ is the bipartite graph $\ExtGraph_X(w)$ where the vertex set is the disjoint union of $\LExt_X(w)$ and $\RExt_X(w)$, and there is an edge between $a\in\LExt_X(w)$ and $b\in\RExt_X(w)$ if $(a,b)\in \Ext_X(w)$.

A word $w\in\lang(X)$ is said to be \emph{dendric} if $\ExtGraph_X(w)$ is a tree. We naturally extend this terminology to a shift space $X$ when it is true for all $w \in \lang(X)$. 

\begin{example}
Continuing with the example of the Tribonacci shift space $X$, the extension graph $\ExtGraph_X(a)$ is given in below.
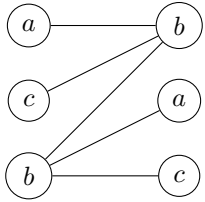
\begin{center}
    \begin{tikzpicture}
    \node[draw,circle] (ra) at (1,0) {$a$};
    \node[draw,circle] (rb) at (1, 1) {$b$};
    \node[draw,circle] (rc) at (1, -1) {$c$};
    
    \node[draw,circle] (la) at (-1,1) {$a$};
    \node[draw,circle] (lb) at (-1, -1) {$b$};
    \node[draw,circle] (lc) at (-1, 0) {$c$};
    
    \draw[-] (la) to node{} (rb);
    \draw[-] (lb) to node{} (ra);
    \draw[-] (lb) to node{} (rb);
    \draw[-] (lb) to node{} (rc);
    \draw[-] (lc) to node{} (rb);
    \end{tikzpicture}
\end{center}
Again, the factors $aab, baa,bab,bac$, and $cab$ can be seen in $abacabaabacababac$. Observe that $\ExtGraph_X(a)$ is a tree, so $a$ is dendric. More generally, the Tribonacci shift space is episturmian, and therefore, dendric~\cite{Berthe2015}.
\end{example}

\section{Proof of the main result}
\label{s:main}

The proof of the converse of the Return Theorem uses techniques and results from symbolic dynamics, group theory, and combinatorics. We start with some results at the intersection between symbolic dynamics and combinatorics.

The first ingredient needed for the proof is the notion of \emph{derived shift space}, which was introduced in \cite{Durand1998} to characterize substitutive sequences. Let $X$ be a minimal shift space and $w\in\lang(X)$. Let $\cB$ be an alphabet with a bijection $\theta_w\from\cB\to \ret_X(w)$, which we naturally extend to maps $\theta_w\from F_\cB \to F_\cA$ and $\theta_w\from\cB^\Z\to A^\Z$ whenever convenient. The \emph{derived shift space} of $X$ with respect to $w$ is the shift space
\begin{equation*}
    D_w(X) = \theta_w^{-1}(X).
\end{equation*}
We call $\theta_w$ the \emph{derivating substitution}. Note that the operation $X\mapsto D_w(X)$ only depends on $w$, up to a relabeling of the alphabet of $D_w(X)$. It is direct to check that $D_w(X)$ is also minimal. Moreover, the derived shift space of a dendric shift space is again dendric by~\cite[Theorem 5.13]{Berthe2015a}. Using ideas from the proof of this result, we obtain the following technical lemma.

\begin{lemma}\label{L:isomorphism generalized extension graph}
    Let $X$ be a minimal shift space and $w \in \lang(X)$. Then $\ExtGraph_X(w)$ is the image of $\ExtGraph_{D_w(X)}(\varepsilon)$ under a graph morphism.
\end{lemma}

\begin{proof}
    Let $\theta_w\from \cB\to\ret_X(w)$ be the derivating substitution used to define $D_w(X)$. We consider the set 
    \[
        \ret'_X(w) = \{r' \in \cA^+ \mid wr'\in \lang(X)\cap \cA^*w\setminus \cA^+w\cA^+\}.
    \]
    of \emph{right return words} to $w$ and we consider $\theta'_w\from \cB^* \to \cA^*$ defined by $\theta'_w(u) = w^{-1}\theta_w(u)w$ (thus it is a bijection between $\cB$ and $\ret_X'(w)$). Observe that $z\in\lang(D_w(X))$ if and only if $\theta_w(z)w\in\lang(X)$, if and only if $w\theta_w'(z)\in\lang(X)$.
    
    Consider the map $\Theta$ defined on the vertices of $\ExtGraph_{D_w(X)}(\varepsilon)$ as follows:
    \begin{itemize}
        \item $r\in\LExt_{D_w(X)}(\varepsilon)$ is mapped to the last letter of $\theta_w(r)$;
        \item $s\in\RExt_{D_w(X)}(\varepsilon)$ is mapped to the first letter of $\theta_w'(s)$.
    \end{itemize}
    Assume that $r\in \LExt_{D_w(X)}(\emptyw)$ and $s\in \RExt_{D_w(X)}(\emptyw)$ are connected by an edge in $\ExtGraph_{D_w(X)}(\emptyw)$; in other words, $rs\in\lang(D_w(X))$. It then follows that $\theta_w(rs)w = \theta_w(r)w\theta_w'(s)\in\lang(X)$, and thus $\Theta(r)w\Theta(s)\in\lang(X)$. This shows that $\Theta$ defines a graph morphism from $\ExtGraph_{D_w(X)}(\emptyw)$ to a subgraph of $\ExtGraph_X(w)$. 
    
    To show that it is onto, assume that $a\in\LExt_X(w)$ and $b\in\RExt_X(w)$ are connected by an edge in $\ExtGraph_{X}(w)$; in other words $awb\in\lang(X)$. By uniform recurrence, there exist words $u$ and $v$ such that $uw$ starts with $w$ and ends with $aw$, $wv$ starts with $wb$ and ends with $w$, and $uwv\in\lang(X)$. Assuming that $u$ and $v$ are of minimal lengths, we get $u\in\ret_X(w)$ and $v\in\ret_X'(w)$, hence there exist $r,s\in\cB$ such that $\theta_w(r)=u$ and $\theta_w'(s)=v$. Moreover, $\theta_w(r)w\theta_w'(s) = \theta_w(rs)w \in \lang(X)$ so $rs \in \lang(D_w(X))$. Considering $r$ as an element of $\LExt_{D_w(X)}(\emptyw)$ and $s$ as an element of $\RExt_{D_w(X)}(\emptyw)$, $r$ and $s$ are connected in $\ExtGraph_{D_w(X)}(\emptyw)$ and satisfy $\Theta(r)=a$ and $\Theta(s)=b$, which proves that $\Theta$ is onto.
\end{proof}

We will also need the following result.

\begin{lemma}[{\cite[Proposition 2.6]{Durand1998}}]
\label{l:durand derivation en chaine}
    Let $X$ be a minimal shift space. If $\theta_w$ is the derivating substitution for $w \in \lang(X)$ used to define $D_w(X)$, then for every $u \in \lang(D_w(X))$
    \begin{equation*}
        \theta_w(\ret_{D_w(X)}(u)) = \ret_X(\theta_w(u)w).
    \end{equation*}
\end{lemma}

\cref{P:free implies connected} below is the key to prove the converse of the Return Theorem. It relies on the previous lemmas and uses the notion of Rauzy graphs. The \emph{order-$m$ Rauzy graph} of a shift space $X$, denoted $\RauzyGraph_m(X)$, is the directed graph whose vertices are the elements of $\lang_m(X)$, and there is an edge from $u$ to $v$ if there are letters $a$ and $b$ such that $ub = av \in \lang(X)$; this edge is labeled by $a$. The label of a path is the concatenation of the labels of the edges. 

\begin{proposition}\label{P:free implies connected}
    Let $X$ be a minimal shift space and $w \in \lang(X)$. If the set $\ret_X(w)$ is a basis of $\gen{\ret_X(w)}$, and if there exists $u \in \ret_X(w)w$ for which $\gen{\ret_X(u)} = \gen{\ret_X(w)}$, then $\ExtGraph_X(w)$ is connected.
\end{proposition}

\begin{proof}
    Let $\theta_w\from \cB \to \ret_X(w)$ be a derivating substitution for $w$ and let $a\in\cB$ be such that $\theta_w(a)w = u$. By \cref{l:durand derivation en chaine}, we have $\gen{\theta_w(\ret_{D_w(X)}(a))} = \gen{\ret_X(u)} = \gen{\ret_X(w)}$. Since $\ret_X(w)$ is a basis of $\gen{\ret_X(w)}$, one can naturally extend $\theta_w$ into a group isomorphism from $F_\cB$ to $\gen{\ret_X(w)}$. It follows that
    \begin{equation*}
        \gen{\ret_{D_w(X)}(a)} = \theta_w^{-1}\gen{\theta_w (\ret_{D_w(X)}(a))} = F_\cB.
    \end{equation*}
    
    Let us show that $\ExtGraph_X(w)$ is connected. Let $H$ be the subgroup of $F_\cB$ generated by the labels of the loops based at the vertex $a$ in $\RauzyGraph_1(D_w(X))$ (this is sometimes called the \emph{Rauzy group} with respect to $a$, detailed definitions may be found in \cite[Section 4]{GouletOuellet2022}). As return words are particular examples of labels of loops, $\gen{\ret_{D_w(X)}(a)}$ is a subgroup of $H$. Hence, $H = F_\cB$. By~\cite[Lemma 7.2]{GouletOuellet2022}, this implies that $\ExtGraph_{D_w(X)}(\emptyw)$ is connected. As, by \cref{L:isomorphism generalized extension graph}, $\ExtGraph_X(w)$ is the image of $\ExtGraph_{D_w(X)}(\emptyw)$ under a graph morphism, this shows that $\ExtGraph_X(w)$ is connected.
\end{proof}

To prove the main result, we need one last combinatorial ingredient: The \emph{multiplicity} (or \emph{bilateral order}~\cite{Balkova2008,Cassaigne1997}) of a word $w\in\lang(X)$ is the quantity
\begin{equation*}
    m_X(w) = \card \Ext_X(w) - \card \LExt_X(w) - \card \RExt_X(w) + 1.
\end{equation*}
A word is \emph{strong} if $m_X(w)>0$, \emph{neutral} if $m_X(w)=0$, and \emph{weak} if $m_X(w)<0$.

Multiplicity is related to the \emph{factor complexity} $p_X(n) = \card \left(\lang(X) \cap \cA^n\right)$ as follows (see \cite[Proposition 3.5]{Cassaigne1997}): If $s_X(n) = p_X(n+1)-p_X(n)$ and $b_X(n) = s_X(n+1)-s_X(n)$, then
\begin{equation}\label{Eq:second diff of complexity}
    b_X(n) = \sum_{w\in \lang(X)\cap \cA^n}m_X(w).
\end{equation}

Multiplicity is also related to dendricity. Indeed, if we say that a word $w \in \lang(X)$ is \emph{connected} whenever the graph $\ExtGraph_X(w)$ is connected, then we have the following simple observations.

\begin{lemma}\label{L:connected and neutral}
    Let $X$ be a shift space and $w \in \lang(X)$.
    \begin{enumerate}
        \item If $w$ is connected, then $w$ is strong or neutral.
        \item If $w$ is connected and neutral, then $w$ is dendric.
    \end{enumerate}
\end{lemma}

\begin{proof}
    This directly follows from the observation that $m_X(w) - 1$ is the difference between the number of edges and the number of vertices in $\ExtGraph_X(w)$. By a classical graph theory result (see \cite[Exercise 2.1.5 and Corollary 2.4.2]{book/Bondy1976} for example), if $w$ is connected then $m_X(w) \geq 0$, with equality if and only if $\ExtGraph_X(w)$ is a tree.
\end{proof}

We can now prove the converse of the Return Theorem.

\begin{theorem}\label{t:main}
    Let $X$ be a minimal shift space over $\cA$ such that, for all $w\in \lang(X)$, $\ret_X(w)$ is a basis of the free group $F_\cA$. Then $X$ is dendric.
\end{theorem}

\begin{proof}
    We first observe that any $w \in \lang(X)$ is connected by \cref{P:free implies connected}. In particular, $\lang(X)$ has no weak factor by \cref{L:connected and neutral}. Moreover, all the return sets have the same cardinality $\card\cA$ by assumption, therefore, the factor complexity of $X$ is given by $p_X(n) = (\card\cA - 1)n + 1$ for all $n$ by \cite[Theorem~4.5]{Balkova2008}. It follows that $s_X(n) = \#\cA-1$ and $b_X(n) = 0$. Using again the fact that $X$ has no weak factor, this implies, by \cref{Eq:second diff of complexity}, that $m_X(w) = 0$ for all $w \in \lang(X)$, and so $X$ is dendric by \cref{L:connected and neutral}.
\end{proof}

Using \cite[Theorem~5.19]{Berthe2015a}, we can strengthen the Return Theorem by observing that, for any factor $w$ in a dendric shift space $X$, $\ret_X(w)$ is a tame basis of $F_\cA$. A basis is said to be {\em tame} if it is obtained by applying to $\cA$ a composition of permutations and morphisms of the form $\alpha_{a,b}$ and $\tilde{\alpha}_{a,b}$ where, for $a \ne b$ in $\cA$,
\[
    \alpha_{a,b}(c) =
    \begin{cases}
        ab, & \text{ if  $c=a$;} \\
        c, & \text{ otherwise}
    \end{cases}
    \quad
    \text{and}
    \quad
    \tilde{\alpha}_{a,b}(c) =
    \begin{cases}
        ba, & \text{ if $c=a$;} \\
        c, & \text{ otherwise.}
    \end{cases}
\]

The following theorem synthesizes \cref{t:main} together with \cite[Theorem~5.19]{Berthe2015a} and the Return Theorem.

\characThm*

\bibliographystyle{abbrvnat}
\bibliography{biblio}

\end{document}